\documentclass[10pt]{article}

\usepackage{amssymb}
\usepackage{amsfonts}
\usepackage{amsmath}
\usepackage{amscd}
\usepackage[all,cmtip]{xy}
\usepackage{amsthm}
\usepackage{setspace}
\usepackage{enumerate}
\usepackage{graphicx}
\usepackage{tikz}
\usepackage{verbatim}
\usepackage{geometry}

\usepackage{todonotes}

\theoremstyle{plain}

\newtheorem{theorem}{Theorem}
\newtheorem*{rigiditylemma}{Rigidity Lemma}
\newtheorem*{thm*}{Theorem}

\newtheorem*{lem*}{Lemma}

\theoremstyle{definition}

\newtheorem*{example*}{Example}

\newtheorem*{remark}{Remark}

\newcommand{\CC}{{\mathbb C}}

\newcommand{\ZZ}{{\mathbb Z}}

\usepackage{hyperref}

\title{Rigidity of Kac-Schwarz operators}
\author{Martin T. Luu}

\date{}

\begin{document}

\newcommand{\Addresses}{{
\bigskip
\footnotesize

M. Luu, \textsc{Department of Mathematics, University of California, Davis} \par \nopagebreak \textit{E-mail address:} \texttt{mluu@math.ucdavis.edu}

}}

\maketitle

\begin{abstract}
In his work on the mathematical formulation of 2d quantum gravity Schwarz established a rigidity result for Kac-Schwarz operators for the $n$-KdV hierarchies. Later on, Adler and van Moerbeke as well as Fastr\'{e} obtained different proofs of this result. We give yet another proof of the rigidity, one that in fact holds for all Drinfeld-Sokolov hierarchies.   
\end{abstract}

\section{Introduction}
\label{intro}

In \cite{WIT} Witten conjectures the equivalance of various descriptions of 2d quantum gravity. The conjecture is proven by Kontsevich in \cite{KON}. A central part of these works is to show that the partition function $Z$ of topological gravity corresponds to a tau function of the KdV integrable hierarchy. The function $Z$ is singled out within the KdV phase space by the string equation. Concretely, there is a unique formal power series $\tau$ in $\CC[\![t_{1},t_{3},\cdots]\!]$, the Witten-Kontsevich tau function, such that
\begin{eqnarray}
\label{string-equation}
\left ( \sum_{\substack{i >0 \\ \\ i \textrm{ odd }}} \left ( \frac{i+2}{2} \cdot t_{i+2}-\delta_{i,1} \right)\partial_{t_{i}} +\frac{1}{4}\cdot t_{1}^{2} \right )\tau = 0
\end{eqnarray}
and $\tau^{2}$ is equal to $Z$.  

Kac and Schwarz give in \cite{KS} a geometric description of this $\tau$. The totality of all KdV tau functions can geometrically be described by the Sato Grassmannian. In \cite{KS} it is shown how to single out the point $V$ in this Grassmannian corresponding to the Witten-Kontsevich tau function. In terms of an indeterminate $z$, the points of the Grassmannian correspond to certain complex subspaces of $\CC(\!(1/z)\!)$. The point $V$ corresponding to Equation (\ref{string-equation}) can be described via the stabilization conditions 
\begin{eqnarray}
\label{KdV-equation}
z^{2}V\subseteq V\\
\label{Kac-Schwarz-equation}
\left ( \partial_{z^{2}} - \frac{1}{4z^{2}}+z \right )V \subseteq V
\end{eqnarray}
The stabilization condition in Equation (\ref{KdV-equation}) simply means that $V$ does not correspond to a general KP tau function, but rather happens to be a KdV tau function. On the other hand, Equation (\ref{Kac-Schwarz-equation}) is very restrictive: It is known that there is exactly one point $V$ satisfying it in addition to Equation (\ref{KdV-equation}). The differential operator stabilizing $V$ in Equation (\ref{Kac-Schwarz-equation}) is called a Kac-Schwarz operator. 

Witten's conjecture and Kontsevich's proof thereof create an intriguing connection between quantum field theory and integrable systems. The KdV hierarchy is but one member of the infinite collection of Drinfeld-Sokolov hierarchies: Fixing a complex finite-dimensional simple Lie algebra $\mathfrak g$ and corresponding (possibly twisted) affine Lie algebra $\mathfrak g^{(k)}$ ($k \in \{1,2,3\}$), Drinfeld and Sokolov construct in \cite{DS} a Lie theoretic generalization for $\mathfrak g^{(k)}$ of the KdV integrable system. The latter corresponds to the special case of $\mathfrak s\mathfrak l_{2}^{(1)}$. Generalizations of Witten's conjecture for these more general hierarchies have been intensely studied. See the work of Fan, Jarvis, Ruan \cite{FJR} for the case where $\mathfrak g$ is of ADE type and the work of Liu, Ruan, Zhang \cite{LRZ} for the BCFG types.

In the present work we are interested in properties of Kac-Schwarz operators in the generality of arbitrary affine Lie algebras. To explain our result, we first come back to the KdV Kac-Schwarz operator of Equation (\ref{Kac-Schwarz-equation}). One can let the point $V$ flow within the phase space along the KdV flows and this corresponds to adding to the Kac-Schwarz operator an arbitrary element of $\CC(\!(z^{-1})\!)$ with no powers of $z$ less than $-1$ occurring. Furthermore, there is a notion of gauge equivalent points of the Grassmannian, and it turns out that moving $V$ within its gauge equivalence class has the effect of adding to the Kac-Schwarz operator an element of $\CC(\!(z^{-1})\!)$ with no powers of $z$ higher than $-3$ occurring. Therefore the term $-1/4z^{2}$ in the Kac-Schwarz operator of Equation (\ref{Kac-Schwarz-equation}) plays a very special role: It is the unique part that can not be adjusted via the flows of the integrable system or via gauge transformations. In fact, more is true: Schwarz shows in \cite{SCH} that the factor $-1/4$ of the $z^{-2}$ term is the unique value at which the stabilization conditions of Equation (\ref{KdV-equation}) and Equation (\ref{Kac-Schwarz-equation}) admit a solution $V$. This can be viewed as a rigidity result for Kac-Schwarz operators. By now, there are at least two additional proofs of this result, see the work of Adler and van Moerbeke \cite{AVM}, and Fastr\'{e} \cite{FAS}. 

The starting point of our considerations is the observation by Kac and Schwarz \cite{KS} that when the KdV hierarchy is viewed as the $\mathfrak s \mathfrak l_{2}^{(1)}$ Drinfeld-Sokolov hierarchy then the special factor $-1/4$ attains clear Lie theoretic meaning. From this perspective, the rigidity of the Kac-Schwarz operator is an interesting relation between Lie theory and integrable systems and quantum field theory. We show that this continues to hold beyond the KdV case: In Theorem \ref{main-theorem} we prove the analogue of Schwarz's rigidity result in the generality of all Drinfeld-Sokolov hierarchies, untwisted as well as twisted. 

There are some important remarks concerning the relation between the original scalar formulation of the rigidity and the corresponding matrix version in the Drinfeld-Sokolov context. To give details we first loosely describe Theorem \ref{main-theorem} in the case of untwisted affine Lie algebras (we refer to the main text for all precise definitions). Let $\mathfrak g$ be for simplicity a simple finite-dimensional complex Lie algebra of classical type, with associated untwisted affine Lie algebra $\mathfrak g^{(1)}$, and let $n$ be the dimension of the defining vector representation $V$ of $\mathfrak g$. In terms of a choice of positive Chevalley generators $e_{i}$ of $\mathfrak g$ let $\rho^{\vee}$ be the element of $\mathfrak g$ satisfying $[\rho^{\vee},e_{i}]=1$ for all $i$. Let $\Lambda_{1,z}$ be a cyclic element of the affine Lie algebra $\mathfrak g^{(1)}$ associated to $\mathfrak g$. The big cell $Gr_{n}(\mathfrak g^{(1)})$ of the vector Sato Grassmannian consists of certain complex subspaces of $V(\!(1/z)\!)$. The essence of our main rigidity result, Theorem \ref{main-theorem}, is captured by the following special case. Let $h$ be the Coxeter number of $\mathfrak g$. If $V$ is a point of $Gr_{n}(\mathfrak g^{(1)})$ that in addition satisfies
\begin{eqnarray}
\label{deformation-equation}
\left ( \partial_{z} + \frac{\rho^{\vee} + *\cdot \textrm{id}_{h}}{hz} +\Lambda_{1,z} \right ) V \subseteq V
\end{eqnarray}
for some scalar $*$, then in fact $*=0$. This can be viewed as a rigidity with respect to deforming the Weyl vector $\rho^{\vee}$. As formulated above, the scalar deformation that is ruled out by Theorem \ref{main-theorem} might seem unnaturally restrictive. Let us make three remarks about this point.   
\begin{enumerate}
\item
If $\mathfrak g = \mathfrak s \mathfrak l_{h}$ then the Drinfeld-Sokolov hierarchy has a formulation in terms of a scalar version of the Grassmannian. In this latter description the rigidity of the  $*$-deformation of $\rho^{\vee}$ corresponds precisely to the original rigidity result of Schwarz: If a certain operator
$$\partial_{z^{h}} + \sum_{i \ll \infty} a_{i}z^{i}$$
stabilizes a point of the scalar Sato Grassmannian, then the value of the coefficient $a_{-h}$ is uniquely determined.
\item
One might ask if there is rigidity of $\rho^{\vee}$ with respect to deformations inside the corresponding Cartan subalgebra $\mathfrak h$ of $\mathfrak g$. This type of rigidity does not hold, as we will show: For every element $H$ in $\mathfrak h$ there is in fact a point of the Grassmannian stabilized by
$$\partial_{z} + \frac{H}{z} + \Lambda_{1,z}$$  
\item
The special role played by $\rho^{\vee}$ in this context comes from the relation between stabilization conditions like in Equation (\ref{deformation-equation}) and Virasoro constraints satisfied by the tau function associated to $V$. Within the original type A context this is described in the foundational work of Kac and Schwarz \cite{KS}. For general Drinfeld-Sokolov hierarchies this is treated by Cafasso and Wu in \cite{CW2}. While our Theorem \ref{main-theorem} could be generalized to include rigidity results for more general Lie algebra elements, we restrict to the $\rho^{\vee}$ case since this is the case directly relevant for Witten-Kontsevich points of Drinfeld-Sokolov hierarchies. We refer to \cite{CW2} for more details on the string equation within this general context.
\end{enumerate}

\section{Rigidity of Kac-Schwarz operators}
In Section \ref{scalar-section} we describe the rigidity result for Kac-Schwarz operators in the classical scalar setting associated to the $h$-reduction of the KP hierarchy. In Section \ref{Drinfeld-Sokolov-section} we first recall two reformulations, due to Kac and Schwarz, of this rigidity in terms of the affine Lie algebra $\mathfrak s \mathfrak l_{h}^{(1)}$. This motivates the formulation, proved in Theorem \ref{main-theorem}, of a generalization of the rigidity  to arbitrary Drinfeld-Sokolov hierarchies.

\subsection{The scalar setting}
\label{scalar-section}
Fix an indeterminate $z$ and let $Gr^{\textrm{sc}}$ be the big cell of the index zero part of the scalar Sato Grassmannian. As a set, it simply consists of complex subspaces $V$ of $\CC(\!(1/z)\!)$ whose projection to $\CC[z]$ is an isomorphism. Each point in this Grassmannian corresponds to a Lax operator of the KP hierarchy. Fix a positive integer $h$ and suppose $V$ is a point of $Gr^{\textrm{sc}}$ satisfying the stabilization conditions
$$z^{h} V \subseteq V$$
$$\left (\frac{1}{hz^{h-1}}\cdot \partial_{z} + f(z)\right ) V \subseteq V$$
for some $f(z)$ in $\CC(\!(1/z)\!)$. The first condition implies that $V$ corresponds to the $h$-reduced part of the KP hierarchy and the second condition generalizes Equation (\ref{Kac-Schwarz-equation}) given in the introduction. 

In fact, only a very small part of $f$ cannot be eliminated via gauge transformations or via KP flows: Decompose
$$f = f_{<-h}+f_{-h}+f_{>-h}$$
with
$$f_{<-h}=\sum_{i<-h} c_{i}z^{i} \;\; , \;\; f_{>-h}=\sum_{i>-h} c_{i}z^{i} \;\; , \;\; f_{-h}=c_{-h}\cdot z^{-h}$$
Consider a gauge transformation of the form $\gamma =1 + \sum_{i<0} d_i z^{i}$ where the $d_{i}$ are constants. It changes the Kac-Schwarz operator via
\begin{eqnarray*}
\frac{1}{hz^{h-1}}\cdot \partial_{z} + f(z) \; &\mapsto & \; \frac{1}{hz^{h-1}}\cdot \partial_{z} + f(z) + \frac{1}{hz^{h-1}}\gamma^{-1}\partial_{z}\gamma \\
&&\\
 = \frac{1}{hz^{h-1}}\cdot \partial_{z} + f(z) &+& \frac{1}{hz^{h-1}} \left ( -d_{-1}z^{-2}+  (d_{-1}^{2}-2d_{-2})z^{-3}+ \cdots \right )
\end{eqnarray*}
By recursion with respect to degree, one sees that for any choice of $f_{<-h}$ there is a $\gamma$ such that
$$\gamma^{-1}=hz^{h-1}( \partial_{z}\gamma)^{-1} \cdot f_{<-h}$$
One deduces that gauge transformations correspond to changing $f_{<-h}$. In particular, one can gauge fix $f_{<-h}$ to be $0$.

Let $t_{1},t_{2},\cdots$ denote the KP time variables. The KP flows on the Grassmannian are given by
$$V \mapsto \exp \left( \sum_{i>0} t_{i}z^{i} \right ) V$$
Hence the Kac-Schwarz operator changes as
$$\frac{1}{hz^{h-1}}\cdot \partial_{z}  +f(z) \mapsto \frac{1}{hz^{h-1}}\cdot \partial_{z} + f(z) - \sum_{i>0}\frac{i}{h} t_{i} z^{i-h} $$
Hence (if the KP time variables are viewed as actual complex parameters) the $f_{>-h}$ term of the Kac-Schwarz operator corresponds exactly to the KP time flows. To sum up, $f_{<-h}$ is a gauge term and $f_{>-h}$ is a flow term, so $f_{-h}=c_{-h}\cdot z^{-h}$ is the only remaining term. It turns out that this summand is rigid: Using a Wronskian argument Schwarz shows in \cite{SCH} that there is a unique value of $c_{-h}$ such that $V$ exists. More precisely:
\begin{rigiditylemma}[Version I]
Suppose $V$ is a point in the Sato Grassmannian $Gr^{\textrm{\emph{sc}}}$ such that 
$$z^{h}V\subseteq V \;\;\; ,\;\;\; \left (\frac{1}{hz^{h-1}}\cdot \partial_{z} + \sum_{i\le q}c_{i}z^{i} \right ) V \subseteq V$$
with $q\ge 1$ co-prime to $h$. Then
\begin{eqnarray}
\label{special-value}
c_{-h} = \frac{1-h}{2h}
\end{eqnarray}
\end{rigiditylemma}

After the original proof by Schwarz \cite{SCH}, different proofs of this result were given by Adler and van Moerbeke in \cite{AVM} and also by Fastr\'{e} in \cite{FAS}. In the next section we present another proof that in fact works for all Drinfeld-Sokolov hierarchies (untwisted as well as twisted). The result by Schwarz corresponds to the $\mathfrak s \mathfrak l_{h}^{(1)}$ Drinfeld-Sokolov hierarchy.

\subsection{The Drinfeld-Sokolov setting}
\label{Drinfeld-Sokolov-section}
In this section we generalize the Schwarz rigidity lemma to general Drinfeld-Sokolov hierarchies. The first step is to recall a more Lie theoretic formulation of the lemma. This is based on the important observation by Kac and Schwarz \cite{KS} that there is a clear Lie theoretic meaning to the special value of $c_{-h}$ given in Equation (\ref{special-value}).

To recall this, we first describe the (vector) Sato Grassmannian description of Drinfeld-Sokolov hierarchies. This is developed in detail by Cafasso and Wu in \cite{CW1}. Fix a simple finite-dimensional complex Lie algebra $\mathfrak g$ and consider the corresponding (possibly twisted) affine Lie algebra $\mathfrak g^{(k)}$ (hence $k \in  \{1,2,3\}$ and $k$ is the order of a diagram automorphism of $\mathfrak g$). Let $r$ be the rank of the simple Lie algebra $\mathfrak g$. Let $\textbf{s}=(s_{0},\cdots,s_{r})$ be a collection of $r+1$ non-negative integers, not all zero, and let
\begin{eqnarray}
\label{h-equation}
h_{\textbf{s}}=k\sum_{i=0}^{r} a_{i}s_{i}
\end{eqnarray}
where the $a_{i}$'s are the Kac labels of $\mathfrak g^{(k)}$. 
As discussed in detail in \cite{KAC}, for each choice of $\textbf{s}$ there is an associated concrete realization $\mathfrak g^{(k)}(\textbf{s})$ of $\mathfrak g^{(k)}$ and $\mathfrak g^{(k)}(\textbf{s})\cong \mathfrak g^{(k)}$ for every choice of $\textbf{s}$. For each $\textbf{s}$ one constructs a certain  automorphism $\mu_{\textbf{s}}$ of $\mathfrak g$ of order $h_{\textbf{s}}$. Let $\mathfrak g^{[i]}$ denote the subspace of $\mathfrak g$ on which $\mu_{\textbf{s}}$ acts by multiplication by $\exp(2\pi i \sqrt{-1}/h_{\textbf{s}})$. The type $\textbf{s}$ realization $\mathfrak g^{(k)}(\textbf{s})$ of $\mathfrak g^{(k)}$ is of the form
\begin{eqnarray}
\label{realization-equation}
\mathfrak g^{(k)}(\textbf{s})=\left ( \bigoplus_{i \in \ZZ} \mathfrak g^{[i]} \otimes z^{i}  \right ) \oplus \CC \cdot c \oplus \CC \cdot d
\end{eqnarray}
where $c$ is a central element and $d$ is a certain derivation.
The $\ZZ$-gradation of type $\textbf{s}$ assigns degree $0$ to $c$ and $d$ and degree $i$ to $g\otimes z^{i}$ for $g$ in $\mathfrak g$. In particular, the negative part of $\mathfrak g^{(k)}(\textbf{s})$ is given by
$$\mathfrak g^{(k)}(\textbf{s})_{-} = \bigoplus_{i <0} \mathfrak g^{[i]} \otimes z^{i} $$ 
The Sato Grassmannian approach to the $\mathfrak g^{(k)}$ Drinfeld-Sokolov hierarchy is to let elements in $\mathfrak g^{(k)}(\textbf{s})_{-}$ act on a $\CC(\!(1/z)\!)$-vector space, where $\textbf{s}$ corresponds to the homogeneous gradation, meaning
$$\textbf{s}=\textbf{s}_{\textrm{hom}}= (1,0,\cdots,0)$$

To define this action one first chooses a complex faithful representation of $\mathfrak g$ as in \cite{CW1}. Let $n$ denote its dimension. We identify $\mathfrak g$ with its image under the representation. 

\begin{remark}
\label{trace-remark}
It will be used in the proof of Theorem \ref{main-theorem} that the representation is chosen so that each element of $\mathfrak g$ corresponds to a trace $0$ matrix.
\end{remark}

Let $\mathcal H^{+}_{n} := \CC[z]^{n}$ and let
$$Gr_{n}=\left \{ V \le \mathbb{C}(\!(1/z)\!)^{n} \; \Big | \; z V \subseteq V \textrm{ and } \textrm{proj} : V \rightarrow \mathcal H^{+}_{n} \; \textrm{ is an isomorphism } \right \}$$
Within $Gr_{n}$ sits the big cell of the vector Sato Grassmannian associated to $\mathfrak g^{(k)}$. It is defined as
\begin{eqnarray}
\label{homogeneous-Grassmannian-equation}
Gr_{n}(\mathfrak g^{(k)})=\left \{ V \in Gr_{n} \; \Big | \; V=\exp(A) \mathcal H^{+}_{n}  \;\; \textrm{ for some } A \textrm{ in } \mathfrak g^{(k)}(\textbf{s}_{\textrm{hom}})_{-} \right \}
\end{eqnarray}
where the $\mathfrak g$ action on $\CC(\!(1/z)\!)^{n}$ is via the fixed $n$-dimensional representation of $\mathfrak g$.

Cafasso and Wu define in \cite{CW1} for each $V$ in $Gr_{n}(\mathfrak g^{(k)})$ a tau function $\tau(\textbf{t})$ and it is shown that this Grassmannian tau function agrees up to an explicit non-zero constant with the tau function of a Lax operator of the $\mathfrak g^{(k)}$ Drinfeld-Sokolov hierarchy.  The construction of the latter type of tau function is given by Wu in \cite{WU} and $\exp(A)$ is the dressing operator of the relevant Lax operator.

To describe the Lie theoretic meaning of the special value of $c_{-h}$ in Equation (\ref{special-value}) we now specialize to the case $\mathfrak g = \mathfrak s \mathfrak l_{h}$. Furthermore, for notational reasons, we use the variable $\zeta$ instead of $z$ in the scalar version of the Schwarz rigidity lemma. The point $V$ of the scalar Sato Grassmannian $Gr^{\textrm{sc}}$ in the rigidity lemma then satisfies $\zeta^{h}V\subseteq V$. It is known that this implies that the corresponding tau function is a tau function of the $\mathfrak s \mathfrak l_{h}^{(1)}$ Drinfeld-Sokolov hierarchy. The Grassmannian formulation of this is given in terms of the blending map. To describe this, realize $\mathfrak s \mathfrak l_{h}$ in terms of traceless $h\times h$ matrices and let $e_{i}$ be the transpose of $(0, \cdots,0, 1,0, \cdots,  0)$ where the $1$ is in the $i$'th entry. Let $\zeta$ be an $h$'th root of $z$. The blending map is the isomorphism
$$\xi : \CC(\!(1/\zeta)\!) \longrightarrow \CC(\!(1/z)\!)^{h}$$
given by
$$\zeta^{h-i}f_{i}(\zeta^{h}) \mapsto e_{i}f_{i}(z) \;\; \; \textrm{ where } \;\;\; f_{i} \in \CC(\!(1/\zeta)\!)  \textrm{ and } 1\le i \le h$$
Under this map there is a correspondence of operators that in particular satisfies
$$c_{-h}\cdot \frac{1}{\zeta^{h}} \mapsto M\cdot \frac{1}{hz}:=\begin{pmatrix}
h-1+hc_{-h} &&&&\\
&\ddots &&&\\
&& h-i + hc_{-h}&& \\
&&&\ddots &\\
&&&&hc_{-h}
\end{pmatrix} \cdot \frac{1}{hz} $$

There are two aspects of the Lie theoretic meaning of the result by Schwarz that $c_{-h}=(1-h)/2h$. The first is that it singles out the unique value of $c_{-h}$ for which
$$M \in \mathfrak s \mathfrak l_{h}$$
The second aspect is that when $c_{-h}$ is chosen so that $M$ indeed has trace $0$, then $M$ is in fact a special element of this Lie algebra as we now recall: Let $e_{i,j}$ denote the $h\times h$ matrix with zeros everywhere except a $1$ at the $(i,j)$ entry. Consider the standard choice of Chevalley generators $e_{i,i+1},f_{i+1,i}$ ($1\le i \le h-1$) for $\mathfrak s \mathfrak l_{h}$. Then half the sum of all positive co-roots is the element
$$\rho^{\vee} = \sum_{i=1}^{h} \frac{h+1-2i}{2} \cdot e_{i,i}$$
and one sees that for $c_{-h}=(1-h)/2h$ one obtains exactly 
$$M = \rho^{\vee}$$ 

This yields a more Lie theoretic reformulation of the Schwarz rigidity lemma, an observation due to Kac and Schwarz. To make it precise, let 
\begin{eqnarray}
\label{cyclic-element-definition}
\Lambda_{1,z} = \sum_{i=1}^{h-1}e_{i,i+1} + z \cdot e_{h,1}
\end{eqnarray} 
This is the operator corresponding to multiplication by $\zeta$ under the blending map. Hence under the blending map there is a correspondence   
$$\frac{1}{h\zeta^{h-1}}\cdot \partial_{\zeta} +\frac{1-h}{2h}\frac{1}{\zeta^{h}} +\zeta \;\;\; \mapsto \;\; \;  \partial_{z} +\frac{\rho^{\vee}}{hz} +\Lambda_{1,z}$$ 
between the scalar and matrix formulation of the Kac-Schwarz operator. Therefore the rigidity result of Schwarz, say for $f_{>-h}=z$, can be restated as follows:

\begin{rigiditylemma}[Version II]
\label{Schwarz-rigidity-lemma}
Suppose $V$ is a point of $Gr_{h}(\mathfrak s\mathfrak l_{h}^{(1)})$ and there is a scalar $*$ such that
$$\left( \partial_{z} + \frac{\rho^{\vee}}{h z}  + \frac{*\cdot \textrm{\emph{id}}_{h}}{z}+\Lambda_{1,z}\right ) V \subseteq V$$
Then $*=0$.
\end{rigiditylemma}

\begin{remark}
\label{principal-homogeneous-remark}
The appearance (first described by Kac and Schwarz \cite{KS}) of the operator $ \partial_{z} + \rho^{\vee}/(h z)$
when switching from the scalar to the matrix Lax operator description of the $\mathfrak s \mathfrak l_{h}^{(1)}$ Witten-Kontsevich point is intriguing: 

This operator is the concrete description, in the principal realization of the affine algebra, of one of the members of the Virasoro extension that each affine Lie algebra possesses. This might come as a surprise, since the Sato Grassmannian description of the Drinfeld-Sokolov phase space is based on the standard loop (or homogeneous) realization of the Lie algebra. We elaborate on this special relation between the principal and homogeneous viewpoint for the general Witten-Kontsevich point after the proof of Theorem \ref{main-theorem}.
\end{remark}

To describe the known Lie theoretic meaning of the operator $ \partial_{z} + \rho^{\vee}/(h z)$ we first recall relevant details concerning (possibly twisted) realizations of affine Lie algebras $\mathfrak g^{(k)}$. Consider first the untwisted case $k=1$ and suppose  $\textbf{s}=\textbf{s}_{\textrm{hom}}$ corresponds to the standard loop realization of the affine Lie algebra. In this case the derivation $d$ in Equation (\ref{realization-equation}) can be viewed as the operator $ z\partial_{z}$. This operator fits into an infinite family of derivations $d_{i}=z^{i+1}\partial_{z}$. Analogously, for any choice of $\textbf{s}$  one obtains a collection of derivations $d_{i}^{\textbf{s}}$ of $\mathfrak g^{(k)}(\textbf{s})$ satisfying
$$
\left [d_{i}^{\textbf{s}},d_{j}^{\textbf{s}} \right ]= h_{\textbf{s}} \cdot (j-i) \cdot d_{i+j}^{\textbf{s}}$$
$$\left [d_{i}^{\textbf{s}},g\otimes z^{j} \right ]  =  j \cdot g\otimes z^{ j+ ih_{\textbf{s}}}
$$
where $g$ is in $\mathfrak g$ and $h_{\textbf{s}}$ is as in Equation (\ref{h-equation}). To obtain the usual Witt algebra relations one normalizes
$$\textrm{d}_{i}^{\textbf{s}}:=-\frac{1}{h_{\textbf{s}}} \cdot  d_{i}^{\textbf{s}}$$
For any choice of realization, these derivations can be explicitly described, see for example \cite{WAK}. In particular, for the principal gradation 
$$\textbf{s}_{\textrm{pri}}:=(1,\cdots, 1)$$ 
one has
$$\textrm{d}_{-1}^{\textbf{s}_{\textrm{pri}}} = \partial_{z^{ka_{0}}} + \frac{\rho^{\vee}}{kh} \cdot \frac{1}{z^{ka_{0}}}$$
where $h$ is the Coxeter number of $\mathfrak g$, $\rho^{\vee}$ is half the sum of positive co-roots of $\mathfrak g$ and $\partial_{z^{ka_{0}}}:=(ka_{0} z^{ka_{0}-1})^{-1}\partial_{z}$. In the untwisted case, meaning $k=1$, one has automatically $a_{0}=1$ and $\textrm{d}_{-1}^{\textbf{s}_{\textrm{pri}}}$ is exactly $\partial_{z}+\rho^{\vee}/(hz)$.

Fix a set of Chevalley generators of the finite-dimensional Lie algebra $\mathfrak g$: Denote them by $e_{1},\cdots, e_{r}, f_{1},\cdots, f_{r}$ unless $\mathfrak g^{(k)} = \mathfrak s \mathfrak l_{2n}^{(2)}$ in which case the indexing is chosen to be of the form $e_{0},\cdots,e_{r-1},f_{0},\cdots f_{r-1}$. As described in \cite{KAC}, one can adjoin two elements $e_{0},f_{0}$ in $\mathfrak g^{(k)}$ (or $e_{r},f_{r}$ if $\mathfrak g^{(k)} = \mathfrak s \mathfrak l_{2n}^{(2)}$) to get Chevalley generators of the affine Lie algebra $\mathfrak g^{(k)}$. For example if $k=1$ and $E_{0}$ is a generator of the lowest root space of $\mathfrak g$, then one can take $e_{0}=E_{0}\otimes z$. Let 
\begin{eqnarray}
\label{general-cyclic-element}
\Lambda_{1,z}=\sum_{i=0}^{r} e_{i}
\end{eqnarray}
and note that for $\mathfrak g^{(k)}= \mathfrak s \mathfrak l_{h}^{(1)}$ this formula for $\Lambda_{1,z}$ is consistent with the earlier definition in Equation (\ref{cyclic-element-definition}). Version II of the rigidity result of Schwarz can now be reformulated once more as:
\begin{rigiditylemma}[Version III]
\label{Schwarz-rigidity-lemma-2}
Suppose $V$ is a point of $Gr_{h}(\mathfrak s\mathfrak l_{h}^{(1)})$ and there is a scalar $*$ such that
$$\left(\textrm{\emph{d}}_{-1}^{\textrm{\emph{\textbf{s}}}_{\textrm{\emph{pri}}}}  + \frac{*\cdot \textrm{\emph{id}}_{h}}{z}+\Lambda_{1,z}\right ) V \subseteq V$$
Then $*=0$
\end{rigiditylemma}

This version of the rigidity result can be directly adapted to any affine Lie algebra. We show in Theorem \ref{main-theorem} that it indeed remains true in this generality. We formulate it in somewhat greater generality than the above situation, replacing $\Lambda_{1,z}$ by an arbitrary element in $\mathfrak g^{(k)}$, since the proof is the same. 

Before giving the proof of Theorem \ref{main-theorem} we discuss a point  mentioned in the introduction. It might seem that the above stated rigidity is very narrow, and that maybe there is a rigidity with respect to varying $\rho^{\vee}$ in the Cartan algebra. However this does not hold, as we now explain. We restrict to untwisted algebras $\mathfrak g^{(1)}$ for simplicity, where $\mathfrak g$ is a simple complex Lie algebra with a fixed faithful $n$-dimensional representation as before. 

Cafasso and Wu show in \cite{CW2} (Theorem 3.10) that there is $\gamma = \exp(\sum_{i<0} U_{i})$
(where $U_{i}$ is of principal degree $i$ in $\mathfrak g^{(1)}$) such that
$$\gamma^{-1} \left (\partial_{z} +\frac{\rho^{\vee}}{hz}-\Lambda_{1,z} \right ) \gamma = \partial_{z} -\Lambda_{1,z}$$
Since $\partial_{z}-\Lambda_{1,z}$ maps $\mathcal H_{n}^{+}:=\CC[z]^{n}$ to itself, it follows that $V:= \gamma \mathcal H_{n}^{+}$ satisfies
$$\left (\partial_{z} +\frac{\rho^{\vee}}{hz}-\Lambda_{1,z} \right ) V\subseteq V$$
We claim that a very small modification of the arguments of Cafasso and Wu shows that for any $H$ in the Cartan subalgebra of $\mathfrak g$ there is a point $V$ of the big-cell of the Grassmannian such that, for example, one has
\begin{eqnarray}
\label{Cartan-stabilization-equation}
\left (\partial_{z} +\frac{H}{z}+\Lambda_{1,z} \right ) V \subseteq V
\end{eqnarray}
The proof is the following: First, changing the sign in front of $\Lambda_{1,z}$ does not affect the argument at all, we choose it here to conform with our earlier conventions. The key observation is then that the only point in which $\rho^{\vee}$ enters the proof of \cite{CW2} (Theorem 3.10) is loc. cit. Equation 3.37. This says that there is an element $Y_{-(h+1)}$ of principal degree $-(h+1)$ in the Lie algebra $\mathfrak g^{(1)}$ with
$$[Y_{-(h+1)},\Lambda_{1,z}]=\frac{\rho^{\vee}}{hz}$$
The reason this equation has a solution is that the-right hand side is of principal degree $-h$ and $-h$ is not an exponent of $\mathfrak g^{(1)}$. Hence the right-hand side is in the image of $\textrm{ad } \Lambda_{1,z}$ by \cite{CW2} (Equation 3.33). But one can replace $\rho^{\vee}$ in this argument by any element in the Cartan subalgebra since these elements are all of principal degree $0$. It follows that there is $\mu =  \exp(\sum_{i<0} V_{i})$ (where $V_{i}$ is of principal degree $i$ in $\mathfrak g^{(1)}$) with
$$\mu^{-1} \left (\partial_{z} +\frac{H}{z}+\Lambda_{1,z} \right ) \mu = \partial_{z} +\Lambda_{1,z}$$
Hence one can take $V=\mu \mathcal H_{n}^{+}$ in Equation (\ref{Cartan-stabilization-equation}).

We now come back to the main result.

\begin{theorem}
\label{main-theorem}
Let $\mathfrak g$ be a simple complex Lie algebra with a fixed faithful $n$-dimensional representation, as before. Let $\mathfrak g^{(k)}$ be a corresponding affine Lie algebra with first Kac label $a_{0}$ and let $V$ be a point in $Gr_{n}(\mathfrak g^{(k)})$ such that 
\begin{eqnarray}
\label{stabilization-equation}
\left ( \textrm{\emph{d}}_{-1}^{\textrm{\emph{\textbf{s}}}_{\textrm{\emph{pri}}}} + * \cdot \frac{\textrm{\emph{id}}_{n}}{z^{ka_{0}}} + g \right ) V \subseteq V
\end{eqnarray}
for some $g$ in $\mathfrak g^{(k)}$ and some scalar $*$. Then $*=0$. 
\end{theorem}
\begin{proof}
Suppose $V$ is a point of $Gr_{n}(\mathfrak g^{(k)})$ that satisfies the stabilization condition in Equation (\ref{stabilization-equation}). Write 
$$V = \gamma \cdot \mathcal H^{+}_{n} \;\;\; \textrm{ with } \;\;\; \gamma = \exp(A)  \;\;\; \textrm{ for } \;\;\; A \in \mathfrak g^{(k)}(\textbf{s}_{\textrm{hom}})_{-}$$
Note that $A$ is by definition a particular element of $z^{-1} \mathfrak g[\![z^{-1}]\!]$ and hence one can write
$$\gamma = \textrm{id}_{n} + d_{-1}z^{-1} + d_{-2} z^{-2} + \cdots$$
for some $d_{i}$'s in $\mathfrak g$. The operator
$$
\mathcal S := \gamma^{-1} \left ( \textrm{d}_{-1}^{\textbf{s}_{\textrm{pri}}} + * \cdot \frac{\textrm{id}_{n}}{z^{ka_{0}}} + g  \right ) \gamma$$
is given by
\begin{eqnarray}
\label{new-stabilization-equation}
\partial_{z^{ka_{0}}}  + (ka_{0} z^{ka_{0}-1})^{-1} \gamma^{-1} \cdot \partial_{z}(\gamma) + \gamma^{-1} \left (\frac{\rho^{\vee}}{khz^{ka_{0}}} +  g \right  )\gamma   + * \cdot  \frac{\textrm{id}_{n}}{z^{ka_{0}}} 
\end{eqnarray}
The stabilization condition Equation (\ref{stabilization-equation}) yields 
that this operator maps $\mathcal H_{n}^{+}$ into $\mathcal H_{n}^{+}$. Note that the gauge term is of the form 
\begin{eqnarray}
\label{gauge-term-equation}
(ka_{0} z^{ka_{0}-1})^{-1} \gamma^{-1} \cdot \partial_{z}(\gamma)=\sum_{i \le -1-ka_{0}} e_{i}z^{i}
\end{eqnarray}
for suitable $e_{i}$ in $\mathfrak g$. Let $\textrm{deg}$ denote the $z$-degree gradation on $\mathfrak g \mathfrak l_{n} \; \CC(\!(1/z)\!)$. Via our choice of $n$-dimensional representation of $\mathfrak g$ one has 
an inclusion
$$\bigoplus_{i \in \ZZ} \left ( \mathfrak g^{[i]} \otimes z^{i} \right ) 
\hookrightarrow \mathfrak g \mathfrak l_{n}\; \CC(\!(1/z)\!)$$
of the ``loop part'' of the type $\textbf{s}_{\textrm{hom}}$ realization 
of $\mathfrak g^{(k)}$ given in Equation (\ref{realization-equation}) and the homogeneous gradation is given by the restriction of the degree gradation. In the following equations the subscripts indicate degrees with respect to the indicated gradations of $z$-degree or homogeneous gradation $\textbf{s}_{\textrm{hom}}$. The fact that $\mathcal S \mathcal H_{n}^{+} \subseteq \mathcal H_{n}^{+}$ together with Equation (\ref{new-stabilization-equation}) now yields
$$\left [  (ka_{0} z^{ka_{0}-1})^{-1} \gamma^{-1} \cdot \partial_{z}(\gamma) + \gamma^{-1} \left (\frac{\rho^{\vee}}{khz^{ka_{0}}} + g \right  )\gamma + * \cdot \frac{\textrm{id}_{n}}{z^{ka_{0}}} \right ]_{<0, \textrm{deg}} = 0$$
and in particular
$$\left [ (ka_{0} z^{ka_{0}-1})^{-1} \gamma^{-1} \cdot \partial_{z}(\gamma) + \gamma^{-1} \left (\frac{\rho^{\vee}}{khz^{ka_{0}}} + g \right  )\gamma  + * \cdot \frac{\textrm{id}_{n}}{z^{ka_{0}}}  \right ]_{-ka_{0}, \textrm{deg}} = 0$$
Using Equation (\ref{gauge-term-equation}) it follows that
\begin{eqnarray*}
* \cdot \frac{\textrm{id}_{n}}{z^{ka_{0}}}& =& - \left [(ka_{0} z^{ka_{0}-1})^{-1} \gamma^{-1} \cdot \partial_{z}(\gamma) + \gamma^{-1} \left (\frac{\rho^{\vee}}{khz^{ka_{0}}} + g \right  )\gamma  \right ]_{-ka_{0}, \textrm{deg}} \\
& =& - \left [ \gamma^{-1} \left (\frac{\rho^{\vee}}{khz^{ka_{0}}} + g \right  )\gamma     \right ]_{-ka_{0}, \textrm{deg}} \\
& =& - \left [ \gamma^{-1} \left (\frac{\rho^{\vee}}{khz^{ka_{0}}} + g \right  )\gamma  \right ]_{-ka_{0}, \textbf{s}_{\textrm{hom}}}
\end{eqnarray*}
where the last equation comes from the previously mentioned agreement between the homogeneous gradation and the $z$-degree gradation for elements of $\mathfrak g^{(k)}$. It follows that 
$$* \cdot \frac{\textrm{id}_{n}}{z^{ka_{0}}} \in \mathfrak g^{(k)}(\textbf{s}_{\textrm{hom}})$$ 
Since $\mathfrak g$ is realized as a set of trace zero matrices, see Remark \ref{trace-remark}, it follows that $*=0$ as desired.
\end{proof}
There certainly are points in the Grassmannian satisfying the conditions of the theorem: In complete analogy with the previously discussed KdV case, the Witten-Kontsevich point of the $\mathfrak g^{(k)}$ Drinfeld-Sokolov hierarchy is known to correspond to the point $V$ in $Gr_{n}(\mathfrak g^{(k)})$ such that
\begin{eqnarray}
\label{general-Witten-Kontsevich-equation}
zV\subseteq V \;\; , \;\;
\left ( \textrm{d}_{-1,z}^{\textbf{s}_{\textrm{pri}}} -\Lambda_{1,z}\right ) V \subseteq V
\end{eqnarray}

We refer to \cite{CW2} (Section 3.3) for details (the minus sign in front of $\Lambda_{1,z}$ is simply here in order to follow the normalization choices of loc. cit.). As indicated in Remark \ref{principal-homogeneous-remark}, this Grassmannian description of the Witten-Kontsevich point in terms of the $\textrm{d}_{-1}^{\textbf{s}_{\textrm{pri}}}$ action is interesting, since the way the Grassmannian describes the full Drinfeld-Sokolov phase space is in terms of the homogeneous realization as can be seen from Equation (\ref{homogeneous-Grassmannian-equation}). One can ask what role the homogeneous realization $\textrm{d}_{-1}^{\textbf{s}_{\textrm{hom}}}$ of the very same derivation plays for the Witten-Kontsevich point. Our result in \cite{LUU} can be viewed as an attempt to do so. It turns out, the homogeneous realization does indeed play a role: It corresponds to the non-affine oper of the Witten-Kontsevich point, as we now explain. 

We restrict to the untwisted case $k=1$. Let $L$ be the Drinfeld-Sokolov Lax operator corresponding to $V$ satisfying Equation (\ref{general-Witten-Kontsevich-equation}). We have shown in \cite{LUU} that
$$L=\partial_{x}+\Lambda_{1,z} - x\cdot E_{0}$$
where $E_{0}$ is a suitable generator of the lowest root space of $\mathfrak g$. The corresponding non-affine oper is obtained by setting $z=0$ and is hence given by
$$ \partial_{x} + \sum_{i=1}^{r} e_{i} - x \cdot E_{0} = \textrm{d}_{-1,x}^{\textbf{s}_{\textrm{hom}}} - \tilde \Lambda_{1,x}$$
where $\tilde \Lambda_{1,x}$ is defined like $\Lambda_{1,z}$ in Equation (\ref{general-cyclic-element}) but with $x$ replacing $z$ and the Chevalley generators $e_{i}$ of $\mathfrak g$ switched to $-e_{i}$ and  furthermore $\textrm{d}_{-1,x}^{\textbf{s}_{\textrm{hom}}}$ is defined in the same manner as  $\textrm{d}_{-1}^{\textbf{s}_{\textrm{hom}}}$ but with $x$ playing the role of $z$. This is one way in which the homogeneous derivation plays a crucial role for the Witten-Kontsevich point: The switch from principal to homogeneous description of the derivation $\textrm{d}_{-1}^{\textbf{s}}$ hence essentially corresponds to switching from the Grassmannian symmetry condition to the (non-affine) Lax operator itself, with the important difference of switching the loop variable of the affine algebra $\mathfrak g^{(1)}$ from $z$ to $x$:
$$\textrm{d}_{-1}^{\textbf{s}_{\textrm{pri}}} -\Lambda_{1,z} \; \; \leadsto \; \; \textrm{d}_{-1,x}^{\textbf{s}_{\textrm{hom}}} - \tilde \Lambda_{1,x}$$
Note that the variables $z$ and $x$ are known to often play Fourier dual roles.

Another viewpoint on the homogenous aspect comes from Lemma 3.8 in \cite{CW2}: If $\gamma$ denotes the dressing operator of the Lax operator $L$ it is shown there that
$$\gamma \left ( \textrm{d}_{-1}^{\textbf{s}_{\textrm{pri}}} - \Lambda_{1,z} \right ) \gamma^{-1} = \textrm{d}_{-1}^{\textbf{s}_{\textrm{hom}}} - \Lambda_{1,z} $$
So the dressing operator of the Witten-Kontsevich point gives rise to an isomorphism of connections 
$$\textrm{d}_{-1}^{\textbf{s}_{\textrm{pri}}} - \Lambda_{1,z} \cong \textrm{d}_{-1}^{\textbf{s}_{\textrm{hom}}} - \Lambda_{1,z}$$
on the formal punctured disc $\textrm{Spec } \CC(\!(1/z)\!)$.

\hspace{0.2in}

\textbf{Acknowledgements:}
It is a great pleasure to thank Mattia Cafasso and Albert Schwarz for very helpful exchanges.

\Addresses


\begin{thebibliography}{}
\bibitem[1]{AVM} M. Adler, P. van Moerbecke: A matrix integral solution to two-dimensional $W_{p}$-gravity, Commun. Math. Phys. \textbf{147}, 25 - 56 (1992) 
\bibitem[2]{CW1} M. Cafasso, C. Z. Wu: Tau functions and the limit of block Toeplitz determinants, Int. Math. Res. Not. \textbf{20}, 10339 - 10366 (2015)
\bibitem[3]{CW2} M. Cafasso, C. Z. Wu: Borodin - Okounkov formula, string equation and topological solutions of Drinfeld - Sokolov hierarchies, arXiv:1505.00556
\bibitem[4]{DS} V. Drinfeld, V. Sokolov: Lie algebras and equations of Korteweg - de Vries type, Journal of Soviet Mathematics , 1975 - 2036 (1985)
\bibitem[5]{FAS} J. Fastr\'{e}: A Grassmannian version of the Darboux transformation, Bull. Sci. Math. \textbf{123}, 181 - 232 (1999)
\bibitem[6]{FJR} H. Fan, T. Jarvis, Y. Ruan: The Witten equation, mirror symmetry and quantum singularity theory, Ann. Math. \textbf{178},  1 - 106 (2013)
\bibitem[7]{KAC} V. Kac: Infinite dimensional Lie algebras, third edition, Cambridge Univ. Press (1990)
\bibitem[8]{KON} M. Kontsevich: Intersection theory on the moduli space of curves and the matrix Airy function, Commun. Math. Phys. \textbf{147}, 1 - 23 (1992)
\bibitem[9]{KS} V. Kac, A. Schwarz: Geometric interpretation of the partition function of 2D gravity, Phys. Lett. B \textbf{257}, 329 - 334 (1991)
\bibitem[10]{SCH} A. Schwarz: On solutions to the string equation, Mod. Phys. Lett. A \textbf{6}, 2713 - 2725 (1991)
\bibitem[11]{LRZ} S.-Q. Liu, Y. Ruan, Y. Zhang: BCFG Drinfeld-Sokolov hierarchies and FJRW theory, Invent. Math. \textbf{201}, 711 - 772 (2015)
\bibitem[12]{LUU} M. Luu: Feigin-Frenkel image of Witten-Kontsevich points, preprint
\bibitem[13]{WAK} M. Wakimoto: Affine Lie algebras and the Virasoro algebra I, Japan. J. Math. \textbf{12}, 379 - 400 (1986)
\bibitem[14]{WIT} E. Witten: Two-dimensional gravity and intersection theory on moduli space,
Surv. Diff. Geom. \textbf{1}, 243 - 310 (1991)
\bibitem[15]{WU} C. Z. Wu: Tau functions and Virasoro symmetries for Drinfeld-Sokolov hierarchies, Adv. Math., 603 - 652 (2017) 
\end{thebibliography}
\end{document}